\documentclass[11pt]{amsart}
\usepackage[margin=1in]{geometry}
\usepackage{amsmath,amssymb,amsthm,mathtools}
\usepackage{microtype}
\usepackage[hidelinks,pdfusetitle]{hyperref}
\usepackage{etoolbox}

\makeatletter
\patchcmd{\@settitle}{\uppercasenonmath\@title}{}{}{}
\patchcmd{\@setauthors}{\MakeUppercase{\authors}}{\authors}{}{}
\makeatother
\newtheorem{theorem}{Theorem}[section]
\newtheorem{proposition}[theorem]{Proposition}
\newtheorem{lemma}[theorem]{Lemma}

\theoremstyle{remark}
\newtheorem{remark}[theorem]{Remark}

\newcommand{\supp}{\operatorname{supp}}
\newcommand{\wt}{\operatorname{wt}}
\newcommand{\bwt}{\operatorname{bwt}}
\newcommand{\ones}[1]{\mathbf{1}_{#1}}
\newcommand{\indvec}[1]{\boldsymbol{\chi}_{#1}}

\title{Discrete Unique Continuation on a Simplex}
\author{Linjun Li\textsuperscript{*}}

\address{Department of Mathematics, University of Pennsylvania, PA, USA}
\email{linjun@sas.upenn.edu}
\date{}
\hypersetup{pdfauthor={Linjun Li}}

\begin{document}

\begin{abstract}
For integers $N\ge0$ and $n\ge2$, let
\[
\Delta_N^{(n)}
=\left\{\alpha\in\mathbb Z_{\ge 0}^n:
\alpha_1+\cdots+\alpha_n=N\right\}.
\]
We study discrete unique continuation for functions on this lattice simplex.
For an integer $R\ge1$, let
$g:\Delta_{nR}^{(n)}\to\mathbb R$ satisfy
\[
\sum_{i=1}^n g(\beta+e_i)=0,
\qquad \beta\in\Delta_{nR-1}^{(n)},
\]
where $e_i$ is the $i$th standard basis vector. We prove that nonvanishing
at the center implies
\[
g(R,\ldots,R)\ne0
\quad\Longrightarrow\quad
|\supp g|\ge c_n R^{\lceil n/2\rceil}.
\]
Here $\supp g$ is the set of points where $g$ is nonzero, and $c_n>0$
depends only on $n$. The exponent $\lceil n/2\rceil$ is optimal.
The proof represents the values of $g$ as polynomial coefficients and uses a
\emph{Pascal uncertainty principle}, which bounds from below the total
number of nonzero coefficients of a one-variable polynomial and its unit
translate in terms of their degree.
\end{abstract}

\maketitle
\pagestyle{plain}
\section{Introduction}

Let
\[
\Delta_N^{(n)}
:=\{\alpha\in\mathbb Z_{\ge0}^n:\alpha_1+\cdots+\alpha_n=N\}
\]
be the integer points of the standard simplex. We study functions
$g:\Delta_{nR}^{(n)}\to\mathbb R$ satisfying the relations
\[
\sum_{i=1}^n g(\beta+e_i)=0,
\qquad \beta\in\Delta_{nR-1}^{(n)},
\]
where $n\ge2$, $R\ge1$, and $e_i$ is the $i$th coordinate vector.
For each $\beta$, this is a cancellation relation on the vertices of an
elementary oriented simplex. Our question is how many values of $g$ must
be nonzero if $g(R,\ldots,R)\ne0$.

Theorem~\ref{thm:main} proves that
\[
g(R,\ldots,R)\ne0
\quad\Longrightarrow\quad
|\supp g|\ge c_nR^{\lceil n/2\rceil}.
\]
Here $\supp g=\{\alpha:g(\alpha)\ne0\}$, and $c_n>0$ depends only on $n$.
The exponent is optimal; see Section~\ref{sec:sharpness}.

When $n=2$, the relations form an alternating recurrence, so nonvanishing
at the center implies nonvanishing everywhere. For $n\ge3$, the relevant
predecessor is the triangular-lattice unique-continuation theorem of Li and Zhang
\cite[Theorem~1.9]{LiZhang3D}. It implies the quadratic support bound
when $n=3$, with additional control of magnitudes and allowance for
small errors in the relations. Remark~\ref{rem:n3} gives the precise
identification. The present theorem extends the support conclusion to
all $n$ and shows that the exponent is optimal in each dimension.

The question is motivated by discrete unique continuation and its role
in Anderson--Bernoulli localization. The Anderson
tight-binding model, introduced by Anderson to describe transport in a
disordered lattice \cite{Anderson1958}, is the random Schr\"odinger operator
\[
H_\omega=-\Delta+\delta V_\omega
\quad\text{on }\ell^2(\mathbb Z^D),
\]
where $\delta>0$ is the disorder strength and the on-site values of
$V_\omega$ are independent random variables. We use the convention
\[
(\Delta u)(x)=\sum_{|y-x|_1=1}\bigl(u(y)-u(x)\bigr).
\]
Foundational approaches to localization in arbitrary dimension include
the multiscale analysis of Fr\"ohlich and Spencer \cite{FrohlichSpencer1983}, and the fractional-moment method of Aizenman and
Molchanov \cite{AizenmanMolchanov1993}. A later proof of spectral and dynamical
localization for the one-dimensional lattice Anderson model uses positivity
and large deviations for the Lyapunov exponent \cite{BucajEtAl2019}. Within the multiscale-analysis framework, a crucial
probabilistic input is a Wegner estimate, whose classical density-based form
uses regularity of the single-site distribution \cite{Wegner1981}.

The Bernoulli model is especially delicate because its single-site law has
only two atoms. It has no density, so the usual density-based
spectral-averaging route to a Wegner estimate is not directly available.
Carmona, Klein, and Martinelli proved one-dimensional localization for
Bernoulli and other singular potentials
\cite{CarmonaKleinMartinelli1987}. Damanik, Sims, and Stolz established
localization for one-dimensional continuum Bernoulli--Anderson models
\cite{DamanikSimsStolz2002}. The harmonic-analytic difficulties of the model
also motivated uncertainty-principle questions concerning the simultaneous
concentration of a function and its Fourier transform
\cite{ShubinVakilianWolff1998}. In higher-dimensional continuous space,
Bourgain and Kenig combined quantitative unique continuation with
combinatorial control of free sites to prove Anderson--Bernoulli localization
near the bottom of the spectrum \cite{BourgainKenig2005}. For continuous
Anderson Hamiltonians with singular single-site distributions, Germinet and
Klein established a strong form of localization at the bottom of the spectrum
without assuming regularity of the distribution \cite{GerminetKlein2013}.

% Editorial query: the Logunov--Malinnikova--Nadirashvili--Nazarov
% statement below is retained from the supplied introduction. Its source
% assumes the equation on R^2 and a real measurable potential with |V| <= 1.
Related progress in the continuum comes from Landis' conjecture, which asks how
rapid decay at infinity can force a solution of $\Delta u+Vu=0$ to vanish.
Meshkov exhibited a bounded complex potential and a nonzero solution bounded by
$C\exp(-c|x|^{4/3})$, and showed that the exponent $4/3$ is sharp in the
complex-valued setting \cite{Meshkov1992}. More
recently, Logunov, Malinnikova, Nadirashvili, and Nazarov proved for the real-valued Schr\"odinger equation that
$|u(x)|\leq\exp(-C|x|\sqrt{\log |x|})$ for all $|x|>2$, with $C$ sufficiently
large, forces $u\equiv0$ \cite{LogunovEtAl2025}. 

Passing from the continuum to a lattice introduces a second obstruction: a
literal unique-continuation principle fails because solutions of discrete
Schr\"odinger equations may be supported on lower-dimensional sets. The
polynomial structure of discrete harmonic functions was exploited by
Buhovsky, Logunov, Malinnikova, and Sodin \cite{BLMS2022}. Inspired by their Liouville theorem, Ding and
Smart proved localization near the spectral edge on $\mathbb Z^2$ using a
randomized discrete unique-continuation estimate \cite{DingSmart2020}. Li
proved large-disorder localization for the symmetric
Bernoulli potential on $\mathbb Z^2$ away from small neighborhoods of finitely
many exceptional energies \cite{Li2D2022}. Li and Zhang established
three-dimensional Anderson--Bernoulli localization near the spectral edge,
with a deterministic discrete unique-continuation principle as the main
geometric input \cite{LiZhang3D}; the two- and three-dimensional developments
are also treated in \cite{LiThesis2022}. The quantitative input needed in the Bourgain--Kenig framework must
control both the number of large values of an eigenfunction and their
magnitudes. At an amplitude threshold of order $e^{-CR}$ relative to a
nonzero reference value, the counting estimate sought in a lattice cube
of radius $R$ exceeds $R^{D/2}$ by a positive power. A support estimate
alone does not provide this input, since it gives no lower bound on the
magnitudes of the nonzero values. The recent results described below
address these two aspects of discrete unique continuation.

Here we consider exact cancellation on a simplex of intrinsic dimension
$n-1$. The affine identification in Remark~\ref{rem:n3} connects its
three-variable case with the triangular relation in
\cite[Theorem~1.9]{LiZhang3D}. The following developments extend the
quantitative theory and its applications to lattice Schr\"odinger equations.

After the first version of this paper appeared, Liu, Shi, and Zhang
obtained a quantitative extension of the simplex theorem proved here
\cite{LiuShiZhang2026Simplex}. Their simplex estimate counts values whose
magnitudes are bounded below relative to the central value, retaining
the optimal growth exponent and allowing small errors in the defining
relations \cite[Theorem~2.1]{LiuShiZhang2026Simplex}. They also obtained
deterministic estimates for discrete Schr\"odinger equations with bounded
potentials on $\mathbb Z^D$, $D\ge3$. These lattice estimates give a
nearly sharp count of large values, up to a logarithmic factor, but allow
an amplitude loss exponential in the square of the box radius
\cite[Theorem~2.4]{LiuShiZhang2026Simplex}. They therefore extend the
support theory to quantitative lattice estimates without yet providing
the bounds required for the localization argument.

Subsequently, Li, Liu, and Zhang proved Anderson--Bernoulli localization
near the bottom of the spectrum in every lattice dimension $D\ge2$
\cite{LiLiuZhang2026}. Their main new ingredient is a probabilistic
discrete unique-continuation principle. In dimensions $D\ge3$, a uniform
form of this principle applies to solutions satisfying a prescribed
exponential growth bound relative to a nonzero reference value. With high
probability, the number of values that are large relative to this
reference value is comparable to the box volume up to logarithmic factors,
with an amplitude loss only exponential in the radius. Crucially, the
probability event applies simultaneously to all admissible solutions, so
it can be used for eigenfunctions that depend on the sampled potential
\cite[Theorem~A.1]{LiLiuZhang2026}. A version allowing the potential to
be fixed on a suitably sparse set of sites supplies the
unique-continuation input for their multiscale proof of localization.

We now return to the exact simplex problem and explain the counting
difficulty addressed by the present proof.

A single simplex relation shows that a nonzero value must be accompanied
by at least one other nonzero value. Iterating this observation
does not by itself give the sharp support bound. Different propagation
steps may reach the same points, and the relations do not prescribe
which of the other values must be nonzero. The difficulty is to obtain a
lower bound for the number of distinct supported points from the full
system of relations.

We address this counting problem by applying a polynomial support
inequality associated with the Pascal matrix. For a polynomial $q$, let
$\wt(q)$ denote the number of its nonzero monomial coefficients.
If $q\in\mathbb R[y]$ is nonzero, then
\[
\wt(q(y))+\wt(q(y+1))\ge\deg q+2.
\]
Thus a polynomial of large degree and its unit translate cannot both
have few nonzero coefficients. Translation acts on these coefficients
through a matrix of binomial coefficients, and we call the displayed
inequality the \emph{Pascal uncertainty principle}. It follows from
the classical positivity criterion for binomial determinants
\cite{GesselViennot1985}. Section~\ref{subsec:pascal-method} states the
known criterion precisely and discusses the related uncertainty
estimate of Cohen, Haeupler, and Schulman
\cite{CohenHaeuplerSchulman2018}. Section~\ref{sec:pascal} gives a
lattice-path proof and derives the displayed inequality.

To connect this coefficient inequality with our problem, we encode
$g$ in the homogeneous polynomial
\[
F(z_1,\ldots,z_n)
=
\sum_{\alpha\in\Delta_{nR}^{(n)}}g(\alpha)
\prod_{i=1}^n\frac{z_i^{\alpha_i}}{\alpha_i!}.
\]
The number of nonzero coefficients of $F$ is exactly $|\supp g|$.
The factorials are chosen so that the simplex relations give
$\sum_{i=1}^n\partial_{z_i}F=0$. Consequently, $F$ is unchanged when
the same number is added to all its variables. Subtracting $z_n$ from
every variable gives
\[
F(z_1,\ldots,z_n)
=F(z_1-z_n,\ldots,z_{n-1}-z_n,0).
\]
Thus $F$ is determined by its restriction to $z_n=0$. The coefficients
of that restriction correspond exactly to the values of $g$ on the
coordinate facet $\{\alpha_n=0\}$. Simultaneous translation therefore
relates the interior values of $g$ to its values on this facet.

The application requires a multivariable version of the uncertainty
principle. In one variable, we compare a polynomial with its unit
translate; in several variables, we compare the polynomials obtained
by adding $1$ to any chosen subset of variables. The estimate we use
depends on the exponents of a specified monomial with a nonzero
coefficient. It therefore retains information that would be lost in a
bound depending only on total degree.
We apply this estimate to a polynomial obtained from a coordinate-facet
restriction of $F$. This requires showing how the nonzero coefficient
of $z_1^R\cdots z_n^R$ in $F$ yields suitable coefficient information
for that polynomial. Section~\ref{subsec:tensor-method} states the
multivariable estimate, and Section~\ref{sec:boundary} proves the
required facts about the facet restriction.

The uncertainty inequality counts coefficients in several translated
polynomials, whereas the desired bound concerns values of $g$. These
counts are not immediately interchangeable: the coefficient lists of
arbitrary partial translates need not represent values of $g$. We use two particular
translates that correspond to restrictions of $F$ on two coordinate
facets. Every other translate can be obtained from one of these two
polynomials by translating some of its variables. Expanding each
monomial bounds the number of terms this operation can produce.
This comparison bounds the sum of the translated support sizes in terms
of the support on two facets. The lower bound from the uncertainty
inequality then gives a lower bound for the number of nonzero boundary
values of $g$. Section~\ref{subsec:simplex-method} states the resulting
boundary estimate; Section~\ref{sec:boundary} gives the construction
and comparison.

Finally, we apply the boundary estimate to successive layers of the
simplex. For $0\le t<R$, restrict $g$ to the points with
$\alpha_i\ge t$ for every $i$, and subtract $t$ from every coordinate.
The resulting function is defined on $\Delta_{n(R-t)}^{(n)}$,
satisfies the same simplex relations, and remains nonzero at its center.
Its boundary corresponds to the layer
\[
\{\alpha\in\Delta_{nR}^{(n)}:\min_i\alpha_i=t\}
\]
in the original simplex. These layers are disjoint, so their support
bounds can be added without counting any point twice. Summing over
the scales $R-t$ supplies the final power of $R$ in the support bound.

Section~\ref{sec:main-result} states the theorem, specifies the classical
Pascal input, and explains the estimates used in its application.
Section~\ref{sec:normalization} gives the polynomial normalization. Sections~\ref{sec:pascal} and~\ref{sec:tensor-pascal}
establish the required Pascal inequalities. The application to two
facets and the summation over layers appear in
Sections~\ref{sec:boundary} and~\ref{sec:shells}, respectively.
Section~\ref{sec:sharpness} gives the matching examples.

\section*{Acknowledgments}
This project originated in work begun while the author was a Ph.D. student in the Department of Mathematics at
the University of Pennsylvania. His University of Pennsylvania affiliation reflects the institution at which this work was initiated and a substantial part of it was carried out.
The author acknowledges the use of GPT-5.6 Sol in the human-guided discovery
and exploration of the proof. All arguments were independently verified by
the author.

\section{Main result and the Pascal method}\label{sec:main-result}

\subsection{Notation and main theorem}

For an integer $k\ge1$, write $[k]:=\{1,\ldots,k\}$. If
$\alpha=(\alpha_1,\ldots,\alpha_k)\in\mathbb Z_{\ge0}^k$ and
$z=(z_1,\ldots,z_k)$, we use the multi-index notation
\[
|\alpha|:=\sum_{i=1}^k\alpha_i,
\qquad
\alpha!:=\prod_{i=1}^k\alpha_i!,
\qquad
z^\alpha:=\prod_{i=1}^kz_i^{\alpha_i}.
\]
For any monomial $M$ in the displayed variables, $[M]Q$ denotes the
coefficient of $M$ in a polynomial $Q$; in particular, $[z^\alpha]Q$ is the
coefficient of $z^\alpha$.

For integers $N\ge0$ and $n\ge1$, write
\[
\Delta_N^{(n)}
:=\left\{\alpha\in\mathbb Z_{\ge0}^n:|\alpha|=N\right\}.
\]
Let $e_i$ denote the $i$th standard basis vector. We write
$\ones{k}:=(1,\ldots,1)\in\mathbb R^k$ for the all-ones vector in a specified
dimension $k$. For $S\subseteq[k]$, let $\indvec{S}\in\{0,1\}^k$ be its
indicator vector. For a function on a finite set, write
\[
\supp g:=\{\alpha:g(\alpha)\ne0\}.
\]
If $Q=\sum_\alpha c_\alpha z^\alpha$ is a polynomial, write
\[
\supp Q:=\{\alpha:c_\alpha\ne0\},
\qquad
\wt(Q):=|\supp Q|.
\]
When $N=nR$, we call
\[
R\ones{n}=(R,\ldots,R)\in\Delta_{nR}^{(n)}
\]
the \emph{balanced point}; the corresponding monomial
$z_1^R\cdots z_n^R$ is the \emph{balanced monomial}.

\begin{theorem}[Simplex unique continuation with optimal exponent]\label{thm:main}
Let $n\ge2$ and $R\ge1$ be integers. Suppose that
\[
g:\Delta_{nR}^{(n)}\longrightarrow\mathbb R
\]
satisfies
\begin{equation}\label{eq:simplex-relation}
\sum_{i=1}^n g(\beta+e_i)=0,
\qquad \beta\in\Delta_{nR-1}^{(n)},
\end{equation}
and
\begin{equation}\label{eq:center-nonzero}
g(R,\ldots,R)\ne0.
\end{equation}
For $n\ge3$,
\begin{equation}\label{eq:general-lower-bound}
|\supp g|
\ge
\frac{R^{\lceil n/2\rceil}}
{2^{n-1}(n+1)^{\lfloor (n-2)/2\rfloor}\bigl(\lceil (n-2)/2\rceil+1\bigr)}.
\end{equation}
For $n=2$, one has the stronger exact bound
\[
|\supp g|=2R+1.
\]
In particular, for every fixed $n$ there is a constant $c_n>0$ such that
\[
|\supp g|\ge c_nR^{\lceil n/2\rceil}.
\]
The exponent $\lceil n/2\rceil$ is optimal.
\end{theorem}

\subsection{Pascal matrices and coefficient supports}
\label{subsec:pascal-method}

The Pascal matrix describes how polynomial coefficients change under
translation by one. For an integer $L\ge0$, let
\[
\mathcal P_L=\left(\binom{j}{i}\right)_{0\le i,j\le L},
\qquad \binom{j}{i}=0\quad\text{if }i>j.
\]
The classical binomial-minor criterion of Gessel and Viennot
\cite[Theorem~1 and Corollary~2]{GesselViennot1985} gives the precise
pattern of nonzero minors of this matrix.

\begin{proposition}[Classical binomial-minor criterion]
\label{prop:classical-pascal}
Let $L\ge0$ and $m\ge1$ be integers, and let
\[
0\le i_1<\cdots<i_m\le L,
\qquad
0\le j_1<\cdots<j_m\le L.
\]
Then
\[
\det\left[\binom{j_b}{i_a}\right]_{a,b=1}^m>0
\quad\Longleftrightarrow\quad
i_a\le j_a\quad\text{for every }a\in[m].
\]
If these inequalities fail, the determinant is zero.
\end{proposition}

This criterion supplies the one-variable support estimate used below. To
see the relevant change of basis, write a nonzero polynomial of
degree $L$ in the two bases
\[
q(y)=\sum_{j=0}^L a_jy^j
     =\sum_{i=0}^L b_i(y-1)^i.
\]
Equivalently, $q(y+1)=\sum_i b_i y^i$, so
\[
b_i=\sum_{j=i}^L\binom ji a_j,
\qquad b=\mathcal P_La.
\]
The binomial-minor criterion implies that these two coefficient vectors
cannot both be sparse:
\begin{equation}\label{eq:overview-one-variable}
\wt(q(y))+\wt(q(y+1))\ge L+2.
\end{equation}
We call this consequence the \emph{Pascal uncertainty principle}.
It is sharp, as $q(y)=y^L$ has one nonzero coefficient and its unit
translate has $L+1$. The inequality is also a consequence of the Newton-basis uncertainty
estimate of Cohen, Haeupler, and Schulman
\cite[Lemma~1.4]{CohenHaeuplerSchulman2018}.\footnote{Their statement is as
follows. If a real polynomial $f(x)=\sum_k c_k\binom{x}{k}$ is nonzero with
$s=\#\{k:c_k\ne0\}$, and $a$ is the least nonnegative integer with
$f(a)\ne0$, then $f$ has at most $s-1$ integer zeros in $[a,\infty)$.}
Section~\ref{sec:pascal} recalls the path interpretation of the classical
minors and derives \eqref{eq:overview-one-variable} in
Lemma~\ref{lem:one-dimensional}. Thus the one-variable input is obtained
from a classical determinant property; the issue in the present problem is
how to use this coefficient information on a multidimensional simplex.

\subsection{A multivariable support inequality}
\label{subsec:tensor-method}

We next apply the one-variable inequality successively in each variable,
a procedure we call \emph{tensorization}. For
$p\in\mathbb R[x_1,\ldots,x_d]$ and $S\subseteq[d]$, the polynomial
$p(x+\indvec{S})$ is obtained by replacing $x_i$ with $x_i+1$ for
$i\in S$ and leaving the other variables unchanged. We call these
polynomials the \emph{partial unit translates} of $p$; they include $p$
itself when $S=\varnothing$.

The multivariable estimate keeps track of one nonzero coefficient.
For $d\ge1$ and $\kappa\in\mathbb Z_{\ge0}^d$,
Lemma~\ref{lem:tensor} states that
\begin{equation}\label{eq:overview-tensor}
[x^\kappa]p\ne0
\quad\Longrightarrow\quad
\sum_{S\subseteq[d]}\wt\bigl(p(x+\indvec{S})\bigr)
\ge\prod_{i=1}^d(\kappa_i+2).
\end{equation}
No coordinate of $\kappa$ is required to be a maximal exponent of $p$.
Thus the estimate applies as soon as one suitable nonzero coefficient is
known. A bound expressed only in terms of total degree would not retain
the separate factors on the right-hand side. Even separate degree bounds
in each variable are insufficient. For example,
$p(x_1,x_2)=x_1^L+x_2^L$, with $L\ge1$, has degree $L$ in each variable,
but each partial translate is still a sum of a polynomial in $x_1$ and a
polynomial in $x_2$. Each therefore has at most $2L+1$ monomials, so the
sum of their four support sizes grows only linearly in $L$. The product
bound in \eqref{eq:overview-tensor} uses one nonzero monomial whose
exponents are large in all variables.

The proof must account for cancellation under translation. Substituting
numerical values for the other variables can merge monomials or annihilate
the chosen coefficient, so one cannot simply multiply separate
one-variable estimates. Instead, the proof retains polynomials as
coefficients in the remaining variables and applies induction to their
supports. Section~\ref{sec:tensor-pascal} gives the argument.

\subsection{From translated polynomials to simplex support}
\label{subsec:simplex-method}

We now state the boundary estimate obtained from the multivariable
inequality and explain how it yields the main theorem. We assume $n\ge3$;
the case $n=2$ follows directly from the alternating recurrence.
Let $F$ be the polynomial associated with $g$ in the introduction.
Lemma~\ref{lem:normalization} verifies that its coefficient support is
exactly the support of $g$ and that $\mathcal DF=0$, where
$\mathcal D=\sum_{i=1}^n\partial_{z_i}$.

The boundary estimate will be needed at every smaller scale. Let $r\ge1$
be an integer, and let $H$ be homogeneous of degree $nr$, with
$\mathcal DH=0$ and $[z_1^r\cdots z_n^r]H\ne0$.
Put $d=n-2$ and $\ell_d=\lfloor d/2\rfloor$.
Proposition~\ref{prop:boundary} gives
\begin{equation}\label{eq:overview-boundary}
\bwt(H):=\#\{\alpha\in\supp H:\min_i\alpha_i=0\}
\ge\frac{(r+2)^d}{2^{d+1}(nr+1)^{\ell_d}}.
\end{equation}
Here $\bwt(H)$ is the number of nonzero coefficients on the boundary
of the exponent simplex. In fact, the same lower bound holds for the
support on any prescribed pair of coordinate facets.

The role of the two facets is to control the sum of supports in
\eqref{eq:overview-tensor}. In the polynomial constructed in
Section~\ref{sec:boundary}, the unshifted and fully shifted supports
correspond to these two facets. Each partial translate can be obtained
from the unshifted or fully shifted polynomial by translating at most
$\lfloor d/2\rfloor$ variables. Bounding the number of terms created,
counting the translates, and allowing for overlap of the two facets gives
the denominator in \eqref{eq:overview-boundary}. The construction and this
comparison are carried out together in the proof of
Proposition~\ref{prop:boundary}, after the coefficient correspondences
have been established.

To pass from the boundary to the full support, partition it into layers
\[
L_t:=\{\alpha\in\supp F:\min_i\alpha_i=t\},
\qquad 0\le t\le R.
\]
These sets exhaust the support because every exponent has coordinate
sum $nR$. Differentiating $t$ times in every variable gives
\[
F_t=\partial_{z_1}^t\cdots\partial_{z_n}^tF.
\]
A monomial survives precisely when all its exponents are at least $t$,
and each surviving exponent is reduced by $t$. Distinct monomials remain
distinct, so the layer $L_t$ becomes exactly the boundary support of $F_t$:
\[
\bwt(F_t)=|L_t|.
\]
For $t<R$, the polynomial $F_t$ satisfies the hypotheses of the boundary
estimate at scale $r=R-t$. Applying \eqref{eq:overview-boundary} and
adding the disjoint layers gives
\[
\wt(F)=\sum_{t=0}^R\bwt(F_t)
\ge c_n\sum_{r=1}^R r^{d-\ell_d}
\ge c'_n R^{d-\ell_d+1}
=c'_n R^{\lceil n/2\rceil},
\]
where $c_n,c'_n>0$ depend only on $n$.
Section~\ref{sec:shells} verifies the hypotheses for $F_t$ and carries
out the summation with the explicit constant in Theorem~\ref{thm:main}.
Section~\ref{sec:sharpness} proves optimality.

\subsection{Comparison with the three-variable result}

For completeness, we give the coordinate identification used in the
comparison with Li and Zhang's theorem in the introduction.

\begin{remark}[The three-variable case and the Li--Zhang triangle]
\label{rem:n3}
Assume that $g$ satisfies the hypotheses of Theorem~\ref{thm:main} with
$n=3$, and let
\[
\xi=(-1,0),
\qquad
\eta=\left(\frac12,\frac{\sqrt3}{2}\right),
\qquad
\Psi_R(\alpha)=-(\alpha_2-R)\xi+(\alpha_3-R)\eta.
\]
This is an affine bijection from $\Delta_{3R}^{(3)}$ onto the equilateral
lattice triangle $T_{0;R}$ of \cite[Definition~1.8]{LiZhang3D}; it sends
$(R,R,R)$ to the origin, and its inverse is
\[
\Psi_R^{-1}(s\xi+t\eta)
=(R+s-t,\,R-s,\,R+t).
\]
Define $u:T_{0;R}\to\mathbb R$ by $u(\Psi_R(\alpha)):=g(\alpha)$.
If $a=s\xi+t\eta\in T_{0;\lfloor R/2\rfloor}$, set
$\alpha=\Psi_R^{-1}(a)$. Its first coordinate satisfies
\[
\alpha_1=R+s-t\ge R-\lfloor R/2\rfloor>0,
\]
so $\beta:=\alpha-e_1$ belongs to $\Delta_{3R-1}^{(3)}$. Moreover,
\[
\Psi_R(\beta+e_1)=a,
\qquad
\Psi_R(\beta+e_2)=a-\xi,
\qquad
\Psi_R(\beta+e_3)=a+\eta.
\]
Hence the simplex relation becomes
\[
u(a)+u(a-\xi)+u(a+\eta)=0
\]
on the inner triangle required in \cite[Theorem~1.9]{LiZhang3D}. That theorem
allows an error exponentially small in $R$, relative to $|u(0)|$, and still
forces a constant multiple of $R^2$ values to remain quantitatively large.
It therefore implies the $n=3$ quadratic support bound, with additional
large-value control.
\end{remark}

\section{Factorial normalization}\label{sec:normalization}

We first express the simplex relations as a differential equation for a
polynomial. The factorial weights absorb the coefficients introduced by
differentiation. For $g$ as in Theorem~\ref{thm:main}, define
\begin{equation}\label{eq:F-def}
F(z)=\sum_{\alpha\in\Delta_{nR}^{(n)}}
g(\alpha)\frac{z^\alpha}{\alpha!},
\qquad
\mathcal D=\sum_{i=1}^n\partial_{z_i}.
\end{equation}

\begin{lemma}\label{lem:normalization}
The polynomial $F$ is homogeneous of degree $nR$, satisfies
$\mathcal DF=0$, and has
\[
[z_1^R\cdots z_n^R]F=\frac{g(R,\ldots,R)}{(R!)^n}\ne0,
\qquad
\wt(F)=|\supp g|.
\]
Conversely, every homogeneous polynomial $F$ of degree $nR$ with
$\mathcal DF=0$ defines a function satisfying
\eqref{eq:simplex-relation} by $g(\alpha)=\alpha!\,[z^\alpha]F$.
\end{lemma}

\begin{proof}
For $\beta\in\Delta_{nR-1}^{(n)}$, differentiation gives
\[
[z^\beta](\mathcal DF)
=\sum_{i=1}^n(\beta_i+1)
  \frac{g(\beta+e_i)}{(\beta+e_i)!}
=\frac{1}{\beta!}\sum_{i=1}^n g(\beta+e_i).
\]
Equivalently,
\[
\mathcal DF
=\sum_{\beta\in\Delta_{nR-1}^{(n)}}\frac{z^\beta}{\beta!}
  \sum_{i=1}^n g(\beta+e_i)=0.
\]
Thus $\mathcal DF=0$ is equivalent to the simplex relations.
Homogeneity, the balanced coefficient, and equality of the supports follow
directly from \eqref{eq:F-def}, since all factorials are nonzero.
\end{proof}

For any polynomial $H$ in these variables, the differential equation is
equivalent to invariance under simultaneous translation:
\[
\mathcal DH=0
\quad\Longleftrightarrow\quad
H(z+\lambda\ones{n})=H(z)
\quad\text{for all }z\in\mathbb R^n,\ \lambda\in\mathbb R.
\]
Indeed, the derivative with respect to $\lambda$ is
$(\mathcal DH)(z+\lambda\ones{n})$; the converse follows by differentiating
at $\lambda=0$. We will use this invariance to relate restrictions to
different coordinate hyperplanes in Section~\ref{sec:boundary}.

\section{The one-variable Pascal uncertainty principle}\label{sec:pascal}

For an integer $L\ge0$, translation by $1$ acts on polynomial coefficients
through the Pascal matrix
\[
\mathcal P_L=\left(\binom{j}{i}\right)_{0\le i,j\le L},
\]
where $\binom{j}{i}=0$ for $i>j$. We first establish a rank formula for
its submatrices. The formula follows from the classical positivity
criterion for binomial minors; see
\cite[Theorem~1 and Corollary~2]{GesselViennot1985}. We include the path
argument, since the precise pattern of vanishing minors will be used below.

\begin{lemma}[Pascal submatrix rank]\label{lem:pascal-rank}
Let $I,J\subseteq\{0,\ldots,L\}$, and let $\mu(I,J)$ be the maximum
size of a matching in the bipartite graph with vertex sets $I$ and $J$
and edges $i\sim j$ whenever $i\le j$. Then
\[
\operatorname{rank}\mathcal P_L[I,J]=\mu(I,J).
\]
\end{lemma}

\begin{proof}
A nonzero term in the determinant expansion of a square submatrix gives a
matching between its rows and columns. Consequently,
$\operatorname{rank}\mathcal P_L[I,J]\le\mu(I,J)$. If $\mu(I,J)=0$,
this proves the assertion.

For the reverse inequality, take the endpoint sets of a matching of size
$\rho\ge1$ and write them in increasing order as
\[
i_1<\cdots<i_\rho,
\qquad
j_1<\cdots<j_\rho.
\]
Then $i_a\le j_a$ for every $a$. Indeed, if $i_a>j_a$, the first $a$
columns could be matched only to the first $a-1$ rows. To obtain a
nonzero minor of the required size, we will prove
\begin{equation}\label{eq:pascal-minor-positive}
\det\left[\binom{j_b}{i_a}\right]_{a,b=1}^{\rho}>0.
\end{equation}

We interpret this determinant as a count of disjoint lattice paths.
Consider the finite directed graph whose vertices are the integer points
in
\[
\{(x,y):-L\le x\le0,\ -x\le y\le L\},
\]
with north and east edges whenever both endpoints lie in this set.
Every edge increases $x+y$ by one, so the graph is acyclic. Set
\[
\mathsf s_i=(-i,i),\qquad \mathsf t_j=(0,j).
\]
A path from $\mathsf s_i$ to $\mathsf t_j$ has $i$ east steps and
$j-i$ north steps. All such paths remain in the graph, and hence their
number is $\binom ji$. We allow a path of length zero, so this assertion
also includes $i=j=0$.

The Lindstr\"om--Gessel--Viennot path determinant formula
\cite{Lindstrom1973,GesselViennot1985} states that, in a finite directed
acyclic graph, if $M=(M_{ab})_{a,b=1}^{\rho}$ is the matrix whose entry
$M_{ab}$ counts paths from $s_a$ to $t_b$, with the starting vertices
$s_1,\ldots,s_\rho$ distinct and the terminal vertices
$t_1,\ldots,t_\rho$ distinct, then
\begin{equation}\label{eq:lgv-signed}
\det M=\sum_{\pi\in\mathfrak S_\rho}\operatorname{sgn}(\pi)N_\pi.
\end{equation}
Here $N_\pi$ is the number of families $(P_1,\ldots,P_\rho)$ of
pairwise vertex-disjoint paths with $P_a$ running from $s_a$ to
$t_{\pi(a)}$. Starting and terminal vertices may coincide, and paths of
length zero are included. To recall the proof, expand the determinant as
a signed sum over all path families. Order the vertices so that every
edge goes from an earlier to a later vertex. For any family with an
intersection, choose the earliest shared vertex and the two paths of
smallest indices passing through it. Exchange their portions after that
vertex. This operation preserves the chosen vertex and pair of indices,
so it is an involution; it reverses the sign by transposing the two
terminal vertices. Thus all intersecting families cancel, leaving
\eqref{eq:lgv-signed}.

To apply the formula, we first show that a vertex-disjoint family in our
graph must have $\pi$ equal to the identity. Indeed, suppose
$a<b$ and $\pi(a)>\pi(b)$. Along a path starting at some $\mathsf s_i$,
there is exactly one vertex on each level $x+y=k$ from $k=0$ to its
terminal level. Write $X_a(k)$ and $X_b(k)$ for the first coordinates of
the two paths at level $k$, for $0\le k\le j_{\pi(b)}$. At the initial
level,
\[
X_a(0)=-i_a>-i_b=X_b(0),
\]
whereas at level $j_{\pi(b)}$,
\[
X_a(j_{\pi(b)})\le0=X_b(j_{\pi(b)}).
\]
Each first coordinate increases by either zero or one at a step. The
integer-valued difference $X_a(k)-X_b(k)$ therefore vanishes at some
level, where the two paths share a vertex. This contradicts vertex
disjointness and excludes every nonidentity permutation.

To prove strict positivity, we construct a vertex-disjoint family for the
identity permutation: from $\mathsf s_{i_a}$, first go north to $(-i_a,j_a)$ and
then east to $\mathsf t_{j_a}$. For $a<b$, the vertical portion of the
$b$th path lies strictly to the left of the entire $a$th path, and its
horizontal portion lies strictly above that path. Thus the paths are
pairwise vertex-disjoint, including when one of these portions has
length zero. Formula~\eqref{eq:lgv-signed} now identifies the determinant
in \eqref{eq:pascal-minor-positive} with $N_{\mathrm{id}}\ge1$.
Taking $\rho=\mu(I,J)$ proves the reverse rank inequality.
\end{proof}

\begin{lemma}[Two-shift Pascal uncertainty]\label{lem:one-dimensional}
Let $q\in\mathbb R[y]$ be nonzero of degree $L$. Then
\begin{equation}\label{eq:one-dimensional}
\wt(q(y))+\wt(q(y+1))\ge L+2.
\end{equation}
\end{lemma}

\begin{proof}
Write
\[
q(y)=\sum_{j=0}^L a_jy^j,
\qquad
q(y+1)=\sum_{i=0}^L b_i y^i.
\]
The coefficient vectors satisfy $b=\mathcal P_La$. Let
$J=\{j:a_j\ne0\}$ and $\sigma=|J|$; in particular, $L\in J$.
We claim that $b$ has fewer than $\sigma$ zero coordinates. Suppose
otherwise, and choose a set
$I$ of exactly $\sigma$ such coordinates. Then
\begin{equation}\label{eq:pascal-relation}
\mathcal P_L[I,J]a_J=0.
\end{equation}

A maximum matching between $I$ and $J\setminus\{L\}$ leaves at least
one row unmatched, since there are $\sigma$ rows and only $\sigma-1$
columns. The column $L$ is adjacent to every row, so adjoining it
increases the matching number by exactly one. Lemma~\ref{lem:pascal-rank}
therefore gives
\[
\operatorname{rank}\mathcal P_L[I,J]
=\operatorname{rank}\mathcal P_L[I,J\setminus\{L\}]+1.
\]
The column indexed by $L$ is consequently not in the span of the other
columns. This contradicts \eqref{eq:pascal-relation}, because $a_L\ne0$.
Thus $b$ has at most $\sigma-1$ zero coordinates, and
\[
\wt(q(y+1))\ge L+1-(\sigma-1)=L+2-\wt(q).
\]
\end{proof}

\section{The tensorized Pascal uncertainty principle}\label{sec:tensor-pascal}

We now extend Lemma~\ref{lem:one-dimensional} to several variables,
retaining the information supplied by an arbitrary nonzero coefficient.
Its exponent need not be maximal in any coordinate.
For $\varepsilon\in\{0,1\}^d$, write
\[
(\tau_\varepsilon p)(x):=p(x+\varepsilon).
\]

\begin{lemma}[Tensorized Pascal uncertainty principle]\label{lem:tensor}
Let $d\ge1$, let $p\in\mathbb R[x_1,\ldots,x_d]$, and let
$\kappa\in\mathbb Z_{\ge0}^d$. If $[x^\kappa]p\ne0$, then
\begin{equation}\label{eq:tensor}
\sum_{\varepsilon\in\{0,1\}^d}\wt(\tau_\varepsilon p)
\ge \prod_{i=1}^d(\kappa_i+2).
\end{equation}
\end{lemma}

\begin{proof}
We induct on $d$. For $d=1$, Lemma~\ref{lem:one-dimensional} gives
\[
\wt(p(x_1))+\wt(p(x_1+1))
\ge\deg p+2\ge\kappa_1+2.
\]
Suppose now that $d\ge2$ and the assertion holds in $d-1$ variables. Write
$x=(x',y)$ and $\kappa=(\kappa',k)$, where $y=x_d$ and $k=\kappa_d$.
The polynomial
\[
q(x'):=[y^k]p(x',y)
\]
satisfies
\begin{equation}\label{eq:pd-coeff}
[x'^{\kappa'}]q=[x^\kappa]p\ne0.
\end{equation}

Fix $\varepsilon'\in\{0,1\}^{d-1}$ and collect terms with the same
$x'$-exponent:
\[
p(x'+\varepsilon',y)=\sum_\omega x'^\omega\varphi_\omega(y).
\]
Here $\omega$ ranges over $\mathbb Z_{\ge0}^{d-1}$, and only finitely many
of the polynomials $\varphi_\omega$ are nonzero. Since translation in $x'$
commutes with coefficient extraction in $y$,
\[
[y^k]\varphi_\omega(y)=[x'^\omega]q(x'+\varepsilon').
\]
Consequently, every $\omega\in\supp(q(x'+\varepsilon'))$ gives a polynomial
$\varphi_\omega$ of degree at least $k$. Lemma~\ref{lem:one-dimensional}
then yields
\[
\wt(\varphi_\omega(y))+\wt(\varphi_\omega(y+1))\ge k+2.
\]

Translation in $y$ leaves the $x'$-exponent unchanged, so terms belonging to
different $\omega$ cannot cancel one another. Thus, for $\delta\in\{0,1\}$,
\[
\wt\bigl(\tau_{(\varepsilon',\delta)}p\bigr)
=\sum_\omega\wt\bigl(\varphi_\omega(y+\delta)\bigr).
\]
Summing over $\delta$ and retaining the terms with
$\omega\in\supp(q(x'+\varepsilon'))$, we obtain
\[
\sum_{\delta\in\{0,1\}}
\wt\bigl(\tau_{(\varepsilon',\delta)}p\bigr)
\ge(k+2)\wt\bigl(q(x'+\varepsilon')\bigr).
\]
Finally, sum over $\varepsilon'$ and apply the induction hypothesis to $q$,
using \eqref{eq:pd-coeff}:
\begin{align*}
\sum_{\varepsilon\in\{0,1\}^d}\wt(\tau_\varepsilon p)
&\ge(k+2)\sum_{\varepsilon'\in\{0,1\}^{d-1}}
\wt\bigl(q(x'+\varepsilon')\bigr)\\
&\ge(k+2)\prod_{i=1}^{d-1}(\kappa_i+2)
=\prod_{i=1}^d(\kappa_i+2).
\end{align*}
\end{proof}

The bound is sharp, as the example
$p(x)=x^\kappa$ gives
\[
\sum_{\varepsilon\in\{0,1\}^d}\wt(\tau_\varepsilon x^\kappa)
=\sum_{\varepsilon\in\{0,1\}^d}
\prod_{i:\,\varepsilon_i=1}(\kappa_i+1)
=\prod_{i=1}^d(\kappa_i+2).
\]

\section{A two-facet boundary lower bound}\label{sec:boundary}

We now use Lemma~\ref{lem:tensor} to bound the number of supported
exponents on the union of two coordinate facets of the simplex.
Setting a polynomial variable equal to zero retains precisely the
coefficients whose corresponding exponent is zero. We will compare
the partial translates in that lemma with the supports of these two
restrictions.

Let $n\ge3$ and $r\ge0$ be integers. Suppose that
$H\in\mathbb R[z_1,\ldots,z_n]$ is homogeneous of degree $nr$ and satisfies
\begin{equation}\label{eq:H-assumptions}
\mathcal D H=0,
\qquad
[z_1^r\cdots z_n^r]H\ne0.
\end{equation}
We write
\[
\bwt(H):=
\#\{\alpha\in\supp H:\min_i\alpha_i=0\}
\]
for the number of supported monomials whose exponents lie on the
boundary of the simplex.

\begin{proposition}[Two-facet estimate]\label{prop:boundary}
Put $d=n-2$ and $\ell_d=\lfloor d/2\rfloor$. For every pair of distinct
indices $i,j\in[n]$,
\begin{equation}\label{eq:boundary-estimate}
\#\{\alpha\in\supp H:\alpha_i=0\text{ or }\alpha_j=0\}
\ge
\frac{(r+2)^d}{2^{d+1}(nr+1)^{\ell_d}}.
\end{equation}
The same lower bound therefore holds for $\bwt(H)$.
\end{proposition}

\begin{proof}
The assumptions are invariant under permutations of the variables, so it
suffices to consider the restrictions obtained by setting $z_n=0$
and $z_{n-1}=0$. Set
\[
E_j:=\{\alpha\in\supp H:\alpha_j=0\},
\qquad j\in\{n-1,n\}.
\]
If $r=0$, then $H$ is a nonzero constant. Thus $|E_n\cup E_{n-1}|=1$,
whereas the right-hand side of \eqref{eq:boundary-estimate} equals $1/2$.
We henceforth assume $r\ge1$.

By the translation identity in Section~\ref{sec:normalization},
\begin{equation}\label{eq:translation-invariance}
H(z+\lambda\ones{n})=H(z),
\qquad z\in\mathbb R^n,\quad\lambda\in\mathbb R,
\end{equation}
so the restriction
\[
B(u_1,\ldots,u_{n-1}):=H(u_1,\ldots,u_{n-1},0)
\]
determines $H$: taking $\lambda=-z_n$ in
\eqref{eq:translation-invariance} gives
\begin{equation}\label{eq:H-from-B}
H(z)=B(z_1-z_n,\ldots,z_{n-1}-z_n).
\end{equation}

We first record what the nonzero balanced coefficient of $H$ implies for
$B$. Write
\[
B(u)=\sum_{|\gamma|=nr}b_\gamma u^\gamma.
\]
A term in this sum can contribute to $z_1^r\cdots z_n^r$ in
\eqref{eq:H-from-B} only if $\gamma_i\ge r$ for every $i<n$. For such a
term, choosing $z_i^r$ in each factor $(z_i-z_n)^{\gamma_i}$ leaves
$\sum_{i<n}(\gamma_i-r)=r$ powers of $z_n$. Hence
\begin{equation}\label{eq:center-expansion}
[z_1^r\cdots z_n^r]H
=(-1)^r
\sum_{\substack{|\gamma|=nr\\\gamma_i\ge r\ (i<n)}}
b_\gamma\prod_{i=1}^{n-1}\binom{\gamma_i}{r}.
\end{equation}
Since the left-hand side is nonzero, there is an exponent $\gamma$ with
\begin{equation}\label{eq:deep-gamma}
b_\gamma\ne0,
\qquad \gamma_i\ge r\quad(1\le i\le n-1).
\end{equation}

We next reduce the number of variables while preserving the coefficient
information just obtained. Let $x=(x_1,\ldots,x_d)$ and set
\[
p(x):=B(x_1,\ldots,x_d,1).
\]
The homogeneity of $B$ is essential here. For a monomial $u^\eta$ of $B$,
the exponent omitted by the substitution $u_{n-1}=1$ is
$\eta_{n-1}=nr-\sum_{i=1}^d\eta_i$. Thus distinct monomials of $B$ give
distinct monomials of $p$, with the same coefficients. In particular,
\begin{equation}\label{eq:dehom-weight}
\wt(p)=\wt(B)=|E_n|,
\end{equation}
and, with $\gamma'=(\gamma_1,\ldots,\gamma_d)$,
\[
[x^{\gamma'}]p=b_\gamma\ne0.
\]

We now identify the other restriction with a translate of $p$.
Equation~\eqref{eq:H-from-B} gives
\[
H(x_1,\ldots,x_d,0,-1)
=B(x_1+1,\ldots,x_d+1,1)
=p(x+\ones{d}).
\]
On the left-hand side, a supported monomial with exponent
$\alpha\in E_{n-1}$ becomes
\[
(-1)^{\alpha_n}[z^\alpha]H\,
x_1^{\alpha_1}\cdots x_d^{\alpha_d}.
\]
Its omitted exponent is determined by
$\alpha_n=nr-\sum_{i=1}^d\alpha_i$. This substitution is therefore also
injective on the supported exponents and changes coefficients only
by signs. We conclude that
\[
\wt\bigl(p(x+\ones{d})\bigr)=|E_{n-1}|.
\]

For $S\subseteq[d]$, write
\[
W_S:=\wt\bigl(p(x+\indvec{S})\bigr).
\]
Lemma~\ref{lem:tensor}, applied to the coefficient $[x^{\gamma'}]p$, and
\eqref{eq:deep-gamma} yield
\begin{equation}\label{eq:all-shifts-lower}
\sum_{S\subseteq[d]}W_S
\ge\prod_{i=1}^d(\gamma_i+2)
\ge(r+2)^d.
\end{equation}
% Retain the label of the same estimate formerly displayed in the overview.
\makeatletter
\begingroup
\protected@edef\@currentlabel{\p@equation\theequation}
\label{eq:overview-shift-sum}
\endgroup
\makeatother
To obtain a support bound on the two facets, it remains to bound this sum
in terms of $W_\varnothing=|E_n|$ and $W_{[d]}=|E_{n-1}|$.

Set $\Lambda_r=nr+1$. Since $p$ has total degree at most $nr$, translating
a monomial $x^\eta$ in $k$ variables produces at most
$\prod(\eta_i+1)\le\Lambda_r^k$ monomials, where the product is over the
translated variables. Summing over the monomials of $p$ gives
\begin{equation}\label{eq:forward-bound}
W_S\le\Lambda_r^{|S|}W_\varnothing.
\end{equation}
Cancellation can only decrease the support, so it does not affect this
upper bound. For the estimate from the other restriction, let
$\widetilde p(x)=p(x+\ones{d})$. Then
\[
p(x+\indvec{S})=\widetilde p(x-\indvec{S^c}),
\qquad S^c=[d]\setminus S.
\]
The polynomial $\widetilde p$ also has total degree at most $nr$.
Applying the same expansion bound to these $d-|S|$ translations by $-1$
therefore gives
\begin{equation}\label{eq:backward-bound}
W_S\le\Lambda_r^{d-|S|}W_{[d]}.
\end{equation}

For every $S$, either $|S|\le\ell_d$ or $d-|S|\le\ell_d$. Choosing the
corresponding bound from \eqref{eq:forward-bound} and
\eqref{eq:backward-bound}, we obtain
\[
W_S\le\Lambda_r^{\ell_d}(W_\varnothing+W_{[d]}).
\]
There are $2^d$ subsets $S$. Summing and using
\eqref{eq:all-shifts-lower} yields
\begin{equation}\label{eq:endpoint-sum}
|E_n|+|E_{n-1}|=W_\varnothing+W_{[d]}
\ge\frac{(r+2)^d}{2^d(nr+1)^{\ell_d}}.
\end{equation}
Each exponent in $E_n\cup E_{n-1}$ is counted at most twice in this sum,
so
\[
|E_n\cup E_{n-1}|\ge\frac{|E_n|+|E_{n-1}|}{2}.
\]
Together with \eqref{eq:endpoint-sum}, this proves
\eqref{eq:boundary-estimate}. Finally, $E_n\cup E_{n-1}$ is contained in
$\{\alpha\in\supp H:\min_i\alpha_i=0\}$, giving the assertion for
$\bwt(H)$.
\end{proof}

\section{Summation over layers and proof of the main theorem}\label{sec:shells}

\begin{proof}[Proof of the lower-bound assertions in
Theorem~\ref{thm:main}]
Suppose first that $n=2$, and put
\[
\zeta_k:=g(k,2R-k),\qquad 0\le k\le2R.
\]
The relation \eqref{eq:simplex-relation}, applied at $(k,2R-1-k)$,
gives $\zeta_{k+1}=-\zeta_k$ for $0\le k<2R$. Since
$\zeta_R\ne0$, all $2R+1$ values are nonzero.

For the rest of the proof, assume $n\ge3$, and let $F$ be given by
\eqref{eq:F-def}. The boundary estimate counts only the monomials with
at least one zero exponent. To count the remaining monomials, we group
them according to their smallest exponent and apply the same estimate
to each group. For $0\le t\le R$, define
\begin{equation}\label{eq:Ft}
F_t:=\partial_{z_1}^t\cdots\partial_{z_n}^tF.
\end{equation}
Writing $F=\sum_\alpha c_\alpha z^\alpha$, we have
\[
F_t
=\sum_{\substack{\alpha\in\supp F\\\alpha_i\ge t\ (i\in[n])}}
c_\alpha\frac{\alpha!}{(\alpha-t\ones{n})!}
z^{\alpha-t\ones{n}}.
\]
The coefficients in this sum are nonzero, and its exponents are distinct.
Thus subtraction of $t\ones{n}$ gives a bijection from
$\{\alpha\in\supp F:\alpha_i\ge t\text{ for all }i\}$ onto
$\supp F_t$. In particular,
\begin{equation}\label{eq:shell-bijection}
\bwt(F_t)
=\#\{\alpha\in\supp F:\min_i\alpha_i=t\}.
\end{equation}
Every exponent in $\supp F$ has coordinate sum $nR$, so its smallest
coordinate belongs to $\{0,\ldots,R\}$. These layers therefore
partition $\supp F$, and
\begin{equation}\label{eq:shell-sum}
\wt(F)=\sum_{t=0}^R\bwt(F_t).
\end{equation}

We check that the boundary estimate applies to every derivative in
\eqref{eq:Ft}. Set $r=R-t$. Then $F_t$ is homogeneous of degree $nr$ and satisfies
$\mathcal D F_t=0$, since $\mathcal D$ commutes with the partial
derivatives in \eqref{eq:Ft}. Moreover,
\begin{equation}\label{eq:Ft-center}
[z_1^r\cdots z_n^r]F_t
=\left(\frac{R!}{r!}\right)^n
[z_1^R\cdots z_n^R]F\ne0.
\end{equation}
Consequently, Proposition~\ref{prop:boundary} applies to each $F_t$.
This includes $t=R$: in that case $F_R$ is a nonzero constant and
$\bwt(F_R)=1$, corresponding to the single exponent $R\ones{n}$.

Put
\[
d=n-2,\qquad
\ell_d=\left\lfloor\frac d2\right\rfloor,\qquad
\theta=d-\ell_d=\left\lceil\frac d2\right\rceil.
\]
For $r\ge1$, Proposition~\ref{prop:boundary} gives
\[
\bwt(F_{R-r})
\ge\frac{(r+2)^d}{2^{d+1}(nr+1)^{\ell_d}}
\ge\frac{r^\theta}{2^{n-1}(n+1)^{\ell_d}},
\]
where we used $r+2\ge r$ and $nr+1\le(n+1)r$. Summing over
$1\le r\le R$ in \eqref{eq:shell-sum}, we obtain
\begin{align*}
\wt(F)
&\ge\frac{1}{2^{n-1}(n+1)^{\ell_d}}
\sum_{r=1}^Rr^\theta\\
&\ge\frac{R^{\theta+1}}
{2^{n-1}(n+1)^{\ell_d}(\theta+1)}.
\end{align*}
The last inequality follows by integrating $x^\theta$ over $[0,R]$.
Since $\theta+1=\lceil n/2\rceil$ and
$\wt(F)=|\supp g|$ by Lemma~\ref{lem:normalization}, this proves
\eqref{eq:general-lower-bound}.
\end{proof}

\section{Optimality of the exponent}\label{sec:sharpness}

For positive quantities depending on $R$, write $A_R\asymp_n B_R$ if
$a_nB_R\le A_R\le b_nB_R$ for some constants $a_n,b_n>0$ depending
only on $n$.

\begin{proposition}[Examples attaining the optimal exponent]\label{prop:sharp}
For all integers $n\ge2$ and $R\ge1$, there exists a homogeneous polynomial
$\Phi_{n,R}\in\ker\mathcal D$ of degree $nR$ such that
\[
[z_1^R\cdots z_n^R]\Phi_{n,R}\ne0
\]
and
\[
\wt(\Phi_{n,R})\asymp_n R^{\lceil n/2\rceil}.
\]
Hence the exponent in Theorem~\ref{thm:main} is optimal.
\end{proposition}

\begin{proof}
Suppose first that $n=2\nu$ is even, and define
\begin{equation}\label{eq:even-example}
\Phi_{n,R}
:=\prod_{j=1}^\nu(z_{2j-1}-z_{2j})^{2R}.
\end{equation}
This polynomial is homogeneous of degree $nR$. Each factor is annihilated
by $\mathcal D$, so the product rule gives $\mathcal D\Phi_{n,R}=0$.
The factors involve disjoint pairs of variables; hence
\[
[z_1^R\cdots z_n^R]\Phi_{n,R}
=(-1)^{\nu R}\binom{2R}{R}^{\nu}\ne0,
\qquad
\wt(\Phi_{n,R})=(2R+1)^\nu.
\]
Thus the support has the required growth order $R^{n/2}$.

Now suppose that $n=2\nu+1$ is odd, and define
\begin{equation}\label{eq:odd-example}
\Phi_{n,R}
:=\left[\prod_{j=1}^{\nu-1}
(z_{2j-1}-z_{2j})^{2R}\right]
(z_{2\nu-1}-z_{2\nu})^R
(z_{2\nu}-z_{2\nu+1})^{2R}.
\end{equation}
When $n=3$, the product in brackets is empty and is understood to be $1$.
The degree is $2R(\nu-1)+3R=nR$, and again every factor is annihilated
by $\mathcal D$.

To compute the coefficient and support contributed by the last three
variables, write
\[
Q(a,b,c):=(a-b)^R(b-c)^{2R}.
\]
For $0\le s\le R$ and $0\le u\le2R$, selecting the term containing
$a^s$ in the first factor and the term containing $b^u$ in the second
gives the monomial
\[
a^s b^{R-s+u}c^{2R-u}
\]
with coefficient
\[
(-1)^{3R-s-u}\binom{R}{s}\binom{2R}{u}.
\]
The first and third exponents determine $s$ and $u$, respectively, so
these $(R+1)(2R+1)$ monomials are distinct and all their coefficients
are nonzero. The exponent triple $(R,R,R)$ corresponds to
$s=u=R$, giving
\[
[a^Rb^Rc^R]Q=(-1)^R\binom{2R}{R}.
\]
The remaining factors in \eqref{eq:odd-example} use disjoint variable
pairs, also disjoint from this three-variable block. It follows that
\[
[z_1^R\cdots z_n^R]\Phi_{n,R}
=(-1)^{\nu R}\binom{2R}{R}^{\nu}\ne0
\]
and
\[
\wt(\Phi_{n,R})
=(2R+1)^{\nu-1}(R+1)(2R+1)
=(R+1)(2R+1)^\nu
\asymp_n R^{\nu+1}.
\]
Since $\nu+1=\lceil n/2\rceil$, this proves the polynomial assertion
in both cases.

Finally, define a function on the simplex by reversing the factorial
normalization:
\[
g^{\mathrm{ex}}_{n,R}(\alpha)
:=\alpha!\,[z^\alpha]\Phi_{n,R},
\qquad \alpha\in\Delta_{nR}^{(n)}.
\]
For every $\beta\in\Delta_{nR-1}^{(n)}$, the identity
$(\beta+e_i)!=(\beta_i+1)\beta!$ gives
\begin{align*}
\sum_{i=1}^n g^{\mathrm{ex}}_{n,R}(\beta+e_i)
&=\beta!\sum_{i=1}^n(\beta_i+1)
[z^{\beta+e_i}]\Phi_{n,R}\\
&=\beta!\,[z^\beta](\mathcal D\Phi_{n,R})=0.
\end{align*}
The function $g^{\mathrm{ex}}_{n,R}$ has a nonzero value at
$R\ones{n}$, and its support has cardinality $\wt(\Phi_{n,R})$.
The examples therefore satisfy all the hypotheses of
Theorem~\ref{thm:main} and show that its exponent cannot be increased.
\end{proof}


\begin{thebibliography}{99}

\bibitem{AizenmanMolchanov1993}
M.~Aizenman and S.~Molchanov,
\emph{Localization at large disorder and at extreme energies: an elementary
derivation},
Comm. Math. Phys. \textbf{157} (1993), no.~2, 245--278.

\bibitem{Anderson1958}
P.~W.~Anderson,
\emph{Absence of diffusion in certain random lattices},
Phys. Rev. \textbf{109} (1958), no.~5, 1492--1505.

\bibitem{BourgainKenig2005}
J.~Bourgain and C.~E.~Kenig,
\emph{On localization in the continuous {Anderson--Bernoulli} model in
higher dimension},
Invent. Math. \textbf{161} (2005), no.~2, 389--426.

\bibitem{BucajEtAl2019}
V.~Bucaj, D.~Damanik, J.~Fillman, V.~Gerbuz, T.~VandenBoom, F.~Wang, and
Z.~Zhang,
\emph{Localization for the one-dimensional {Anderson} model via positivity
and large deviations for the {Lyapunov} exponent},
Trans. Amer. Math. Soc. \textbf{372} (2019), no.~5, 3619--3667.

\bibitem{BLMS2022}
L.~Buhovsky, A.~Logunov, E.~Malinnikova, and M.~Sodin,
\emph{A discrete harmonic function bounded on a large portion of
$\mathbb Z^2$ is constant},
Duke Math. J. \textbf{171} (2022), no.~6, 1349--1378.

\bibitem{CarmonaKleinMartinelli1987}
R.~Carmona, A.~Klein, and F.~Martinelli,
\emph{Anderson localization for Bernoulli and other singular potentials},
Comm. Math. Phys. \textbf{108} (1987), no.~1, 41--66.

\bibitem{CohenHaeuplerSchulman2018}
G.~Cohen, B.~Haeupler, and L.~J.~Schulman,
\emph{Explicit binary tree codes with polylogarithmic size alphabet},
in Proceedings of the 50th Annual ACM SIGACT Symposium on Theory of
Computing (STOC 2018), ACM, New York, 2018, 535--544.

\bibitem{DamanikSimsStolz2002}
D.~Damanik, R.~Sims, and G.~Stolz,
\emph{Localization for one-dimensional, continuum,
{Bernoulli--Anderson} models},
Duke Math. J. \textbf{114} (2002), no.~1, 59--100.

\bibitem{DingSmart2020}
J.~Ding and C.~K.~Smart,
\emph{Localization near the edge for the {Anderson Bernoulli} model on the
two dimensional lattice},
Invent. Math. \textbf{219} (2020), no.~2, 467--506.

\bibitem{Drogin2025}
R.~Drogin,
\emph{Localization of one-dimensional random band matrices},
arXiv:2508.05802 (2025).

\bibitem{ElboimGloriaHernandez2025}
D.~Elboim, A.~Gloria, and F.~Hern\'andez,
\emph{Diffusivity of the {Lorentz} mirror walk in high dimensions},
arXiv:2505.01341 (2025).

\bibitem{FrohlichEtAl1985}
J.~Fr\"ohlich, F.~Martinelli, E.~Scoppola, and T.~Spencer,
\emph{Constructive proof of localization in the Anderson tight binding model},
Comm. Math. Phys. \textbf{101} (1985), no.~1, 21--46.

\bibitem{FrohlichSpencer1983}
J.~Fr\"ohlich and T.~Spencer,
\emph{Absence of diffusion in the Anderson tight binding model for large
disorder or low energy},
Comm. Math. Phys. \textbf{88} (1983), no.~2, 151--184.

\bibitem{GerminetKlein2013}
F.~Germinet and A.~Klein,
\emph{A comprehensive proof of localization for continuous {Anderson}
models with singular random potentials},
J. Eur. Math. Soc. \textbf{15} (2013), no.~1, 53--143.

\bibitem{GesselViennot1985}
I.~Gessel and G.~Viennot,
\emph{Binomial determinants, paths, and hook length formulae},
Adv. Math. \textbf{58} (1985), no.~3, 300--321.

\bibitem{KenigSilvestreWang2015}
C.~E.~Kenig, L.~Silvestre, and J.-N.~Wang,
\emph{On {Landis}' conjecture in the plane},
Comm. Partial Differential Equations \textbf{40} (2015), no.~4, 766--789.

\bibitem{Krymskii2024}
S.~T.~Krymskii,
\emph{On the lowest possible dimension of supports of solutions to the
discrete {Schrodinger} equation},
arXiv:2401.02800 (2024).

\bibitem{KunzSouillard1980}
H.~Kunz and B.~Souillard,
\emph{Sur le spectre des op\'erateurs aux diff\'erences finies al\'eatoires},
Comm. Math. Phys. \textbf{78} (1980), no.~2, 201--246.

\bibitem{li2021manhattan}
L.~Li,
\emph{On the {Manhattan} pinball problem},
Electron. Commun. Probab. \textbf{26} (2021), Paper No.~25, 11 pp.

\bibitem{Li2D2022}
L.~Li,
\emph{{Anderson--Bernoulli} localization at large disorder on the {2D} lattice},
Comm. Math. Phys. \textbf{393} (2022), no.~1, 151--214.

\bibitem{LiThesis2022}
L.~Li,
\emph{Anderson--Bernoulli localization on 2D and 3D lattice},
Ph.D. thesis, University of Pennsylvania, 2022.

\bibitem{LiSupport2026}
L.~Li,
\emph{On support cardinality for the discrete {Schr\"odinger} equation},
Lett. Math. Phys. \textbf{116} (2026), no.~4, Paper No.~90.

\bibitem{LiCylinder2026}
L.~Li,
\emph{Polynomial bound for the localization length of the {Lorentz} mirror
model on the {1D} cylinder},
Electron. Commun. Probab. \textbf{31} (2026), Paper No.~61, 14 pp.

\bibitem{LiLiuZhang2026}
L.~Li, S.~Liu, and L.~Zhang,
\emph{Localization near the edge for the lattice {Anderson--Bernoulli}
model on general dimension},
arXiv:2609.30209 (2026).

\bibitem{LiZhang3D}
L.~Li and L.~Zhang,
\emph{{Anderson--Bernoulli} localization on the three-dimensional lattice and
discrete unique continuation principle},
Duke Math. J. \textbf{171} (2022), no.~2, 327--415.


\bibitem{LiuShiZhang2026Simplex}
S.~Liu, Y.~Shi, and Z.~Zhang,
\emph{Quantitative unique continuation on simplex and $\mathbb{Z}^n$},
arXiv:2608.11602v2 (2026).

\bibitem{Lindstrom1973}
B.~Lindstr\"om,
\emph{On the vector representations of induced matroids},
Bull. London Math. Soc. \textbf{5} (1973), no.~1, 85--90.

\bibitem{LogunovEtAl2025}
A.~Logunov, E.~Malinnikova, N.~Nadirashvili, and F.~Nazarov,
\emph{The {Landis} conjecture on exponential decay},
Invent. Math. \textbf{241} (2025), no.~2, 465--508.

\bibitem{Meshkov1992}
V.~Z.~Meshkov,
\emph{On the possible rate of decay at infinity of solutions of second order
partial differential equations},
Math. USSR-Sb. \textbf{72} (1992), no.~2, 343--361.

\bibitem{ShubinVakilianWolff1998}
C.~Shubin, R.~Vakilian, and T.~Wolff,
\emph{Some harmonic analysis questions suggested by
{Anderson--Bernoulli} models},
Geom. Funct. Anal. \textbf{8} (1998), no.~5, 932--964.

\bibitem{Wegner1981}
F.~Wegner,
\emph{Bounds on the density of states in disordered systems},
Z. Phys. B Condensed Matter \textbf{44} (1981), no.~1--2, 9--15.

\bibitem{YauYin2025}
{\raggedright
H.-T.~Yau and J.~Yin,
\emph{Delocalization of one-dimensional random band matrices},
arXiv:2501.01718 (2025).\par}

\end{thebibliography}
\end{document}